\documentclass[conference]{IEEEtran}

\usepackage{amsfonts,amsbsy,bm,amsmath}
\usepackage[nolist]{acronym}
\usepackage{mathrsfs} 

\usepackage{graphicx,cite,amssymb,mathtools,amsthm}
\usepackage{stfloats}
\usepackage{xcolor}
\usepackage{bbm}
\usepackage{tikz}
\usepackage{pgfplots}
\usetikzlibrary{spy}
\pgfplotsset{compat=1.15}
\usetikzlibrary{arrows,shapes.misc,chains,scopes}
\usetikzlibrary{decorations.pathreplacing}
\usepackage{multicol}
\usepackage{siunitx}
\usepackage{csquotes}
\usepackage{makecell}
\usepackage[ruled,linesnumbered,vlined]{algorithm2e}
\SetNlSty{bfseries}{\color{black}}{}
\SetArgSty{textnormal}

\DeclareMathOperator*{\argmin}{argmin}
\newcommand\constructosum[3]{%
    \begin{tikzpicture}[baseline=(char.base), inner sep=0, outer sep=0]
        \draw (#1,0) circle (#2); 
        \node (char) at (0,0) {$#3\sum$}; 
    \end{tikzpicture}%
}

\newcommand{\osum}{\mathop{\mathchoice
        {\constructosum{-0.3ex}{0.13}{\displaystyle}}
        {\constructosum{-0.3ex}{0.09}{\textstyle}}
        {\constructosum{-0.2ex}{0.06}{\scriptstyle}}
        {\constructosum{-0.15ex}{0.05}{\scriptscriptstyle}}
    }\displaylimits
}

\usepackage{nicefrac}
\begin{document}
	\newcommand{\ensemble}{\mathscr{C}}
\newcommand{\code}{\mathcal{C}}
\newcommand{\vecu}{\boldsymbol{u}}
\newcommand{\veci}{\boldsymbol{i}}
\newcommand{\vecf}{\boldsymbol{f}}
\newcommand{\vecv}{\boldsymbol{v}}
\newcommand{\vecp}{\boldsymbol{p}}
\newcommand{\vecx}{\boldsymbol{x}}
\newcommand{\vecone}{\boldsymbol{1}}
\newcommand{\vecuhat}{\hat{\boldsymbol{u}}}
\newcommand{\vecc}{\boldsymbol{c}}
\newcommand{\vecb}{\boldsymbol{b}}
\newcommand{\vecchat}{\hat{\boldsymbol{c}}}
\newcommand{\vecy}{\boldsymbol{y}}
\newcommand{\vecz}{\boldsymbol{z}}
\newcommand{\B}{\boldsymbol{B}}
\newcommand{\G}{\boldsymbol{G}}
\newcommand{\GK}{\boldsymbol{K}}
\newcommand{\Gsys}{\G_{\mathsf{i},\mathsf{sys}}}
\newcommand{\Gnsys}{\G_{\mathsf{i},\mathsf{nsys}}}
\newcommand{\Per}{\boldsymbol{\Pi}}
\newcommand{\I}{\boldsymbol{I}}
\newcommand{\perP}{\boldsymbol{P}}
\newcommand{\perS}{\boldsymbol{S}}
\newcommand{\GH}[1]{\bm{\mathsf{G}}_{#1}}
\newcommand{\T}[1]{\bm{\mathsf{T}}{(#1)}}

\newcommand{\mi}{\mathrm{I}}
\newcommand{\Prob}{P}
\newcommand{\Z}{\mathrm{Z}}
\newcommand{\SPC}{\mathcal{S}}
\newcommand{\Rep}{\mathcal{R}}
\newcommand{\f}{\mathrm{f}}
\newcommand{\SC}{\mathrm{SC}}
\newcommand{\Q}{\mathrm{Q}}
\newcommand{\MAP}{\mathrm{MAP}}
\newcommand{\E}{\mathrm{E}}
\newcommand{\U}{\mathrm{U}}
\newcommand{\Ham}{\mathrm{H}}

\newcommand{\Amin}{\mathrm{A}_{\mathrm{min}}}
\newcommand{\dmin}{d}
\newcommand{\erfc}{\mathrm{erfc}}

\newcommand{\de}{\mathrm{d}}

\newcommand{\decoRule}{\rule{\textwidth}{.4pt}}

\newcommand{\oleq}[1]{\overset{\text{(#1)}}{\leq}}
\newcommand{\oeq}[1]{\overset{\text{(#1)}}{=}}
\newcommand{\ogeq}[1]{\overset{\text{(#1)}}{\geq}}
\newcommand{\ogeql}[2]{\overset{#1}{\underset{#2}{\gtreqless}}}
\newcommand\numeq[1]%
{\stackrel{\scriptscriptstyle(\mkern-1.5mu#1\mkern-1.5mu)}{=}}

\newcommand{\pscd}{\gl{\Prob(\mathcal{E}_{\SCD})}}
\newcommand{\uED}{\gl{\hat{u}}}
\newcommand{\LED}{\gl{L_i^{(\ED)}}}
\newcommand{\inverson}[1]{\gl{\mathbb{I}\left\{#1\right\}}}
\theoremstyle{definition}
\newtheorem{mydef}{Definition}
\newtheorem{prop}{Proposition}
\newtheorem{theorem}{Theorem}

\newtheorem{lemma}{Lemma}
\newtheorem{remark}{Remark}
\newtheorem{example}{Example}
\newtheorem{definition}{Definition}
\newtheorem{corollary}{Corollary}


\definecolor{lightblue}{rgb}{0,.5,1}
\definecolor{normemph}{rgb}{0,.2,0.6}
\definecolor{supremph}{rgb}{0.6,.2,0.1}
\definecolor{lightpurple}{rgb}{.6,.4,1}
\definecolor{gold}{rgb}{.6,.5,0}
\definecolor{orange}{rgb}{1,0.4,0}
\definecolor{hotpink}{rgb}{1,0,0.5}
\definecolor{newcolor2}{rgb}{.5,.3,.5}
\definecolor{newcolor}{rgb}{0,.3,1}
\definecolor{newcolor3}{rgb}{1,0,.35}
\definecolor{darkgreen1}{rgb}{0, .35, 0}
\definecolor{darkgreen}{rgb}{0, .6, 0}
\definecolor{darkred}{rgb}{.75,0,0}
\definecolor{midgray}{rgb}{.8,0.8,0.8}
\definecolor{darkblue}{rgb}{0,.25,0.6}

\definecolor{lightred}{rgb}{1,0.9,0.9}
\definecolor{lightblue}{rgb}{0.9,0,0.0}
\definecolor{lightpurple}{rgb}{.6,.4,1}
\definecolor{gold}{rgb}{.6,.5,0}
\definecolor{orange}{rgb}{1,0.4,0}
\definecolor{hotpink}{rgb}{1,0,0.5}
\definecolor{darkgreen}{rgb}{0, .6, 0}
\definecolor{darkred}{rgb}{.75,0,0}
\definecolor{darkblue}{rgb}{0,0,0.6}

\definecolor{bgblue}{RGB}{245,243,253}
\definecolor{ttblue}{RGB}{91,194,224}

\definecolor{dark_red}{RGB}{150,0,0}
\definecolor{dark_green}{RGB}{0,150,0}
\definecolor{dark_blue}{RGB}{0,0,150}
\definecolor{dark_pink}{RGB}{80,120,90}

\begin{acronym}
\acro{5G}{the $5$-th generation wireless system}
	\acro{APP}{a-posteriori probability}
	\acro{ARQ}{automated repeat request}
	\acro{ASCL}{adaptive successive cancellation list}
	\acro{ASK}{amplitude-shift keying}
	\acro{AUB}{approximated union bound}
	\acro{AWGN}{additive white Gaussian noise}
	\acro{B-DMC}{binary-input discrete memoryless channel}
	\acro{BEC}{binary erasure channel}
	\acro{BER}{bit error rate}
	\acro{biAWGN}{binary-input additive white Gaussian noise}
	\acro{BLER}{block error rate}
	\acro{bpcu}{bits per channel use}
	\acro{BPSK}{binary phase-shift keying}
	\acro{BRGC}{binary reflected Gray code}
	\acro{BSS}{binary symmetric source}
	\acro{CC}{chase combining}
	\acro{CN}{check node}
	\acro{CRC}{cyclic redundancy check}
	\acro{CSI}{channel state information}
	\acro{DE}{density evolution}
	\acro{DMC}{discrete memoryless channel}
	\acro{DMS}{discrete memoryless source}
	\acro{DSCF}{dynamic successive cancellation flip}
	\acro{eMBB}{enhanced mobile broadband}
	\acro{FER}{frame error rate}
	\acro{FHT}{fast Hadamard transform}
	\acro{GA}{Gaussian approximation}
	\acro{GF}{Galois field}
	\acro{HARQ}{hybrid automated repeat request}
	\acro{i.i.d.}{independent and identically distributed}
	\acro{IF}{incremental freezing}
	\acro{IR}{incremental redundancy}
	\acro{LDPC}{low-density parity-check}
	\acro{LFPE}{length-flexible polar extension}
	\acro{LLR}{log-likelihood ratio}
	\acro{MAP}{maximum-a-posteriori}
	\acro{MC}{Monte Carlo}
	\acro{MLC}{multilevel coding}
	\acro{MLPC}{multilevel polar coding}
	\acro{MLPC}{multilevel polar coding}
	\acro{ML}{maximum-likelihood}
	\acro{MC}{metaconverse}
	\acro{PAC}{polarization-adjusted convolutional}
	\acro{PAT}{pilot-assisted transmission}
	\acro{PCM}{polar-coded modulation}
	\acro{PDF}{probability density function}
	\acro{PE}{polar extension}
	\acro{PMF}{probability mass function}
	\acro{PM}{path metric}
	\acro{PW}{polarization weight}
	\acro{QAM}{quadrature amplitude modulation}
	\acro{QPSK}{quadrature phase-shift keying}
	\acro{QUP}{quasi-uniform puncturing}
	\acro{RCU}{random-coding union}
	\acro{RM}{Reed-Muller}
	\acro{RQUP}{reversal quasi-uniform puncturing}
	\acro{RV}{random variable}
	\acro{SC-Fano}{successive cancellation Fano}
	\acro{SCOS}{successive cancellation ordered search}
	\acro{SCF}{successive cancellation flip}
	\acro{SCL}{Successive cancellation list}
	\acro{SCS}{successive cancellation stack}
	\acro{SC}{successive cancellation}
	\acro{SE}{spectral efficiency}
	\acro{SNR}{signal-to-noise ratio}
	\acro{SP}{set partitioning}
	\acro{UB}{union bound}
	\acro{VN}{variable node}
\end{acronym}
	\title{Complexity-Adaptive Maximum-Likelihood Decoding of Modified $\G_N$-Coset Codes}

	\author{\IEEEauthorblockN{Peihong Yuan and Mustafa Cemil Co\c{s}kun}
	    \IEEEauthorblockA{Institute for Communications Engineering (LNT)\\ Technical University of Munich (TUM) \\
	Email: \{peihong.yuan,mustafa.coskun\}@tum.de
	}
}

	\maketitle

\begin{abstract}
A complexity-adaptive tree search algorithm is proposed for $\G_N$-coset codes that implements maximum-likelihood (ML) decoding by using a successive decoding schedule. The average complexity is close to that of the successive cancellation (SC) decoding for practical error rates when applied to polar codes and short Reed-Muller (RM) codes, e.g., block lengths up to $N=128$. By modifying the algorithm to limit the worst-case complexity, one obtains a near-ML decoder for longer RM codes and their subcodes. Unlike other bit-flip decoders, no outer code is needed to terminate decoding. The algorithm can thus be applied to modified $\G_N$-coset code constructions with dynamic frozen bits. One advantage over sequential decoders is that there is no need to optimize a separate parameter.  
\end{abstract}

\section{Introduction}\label{sec:intro}

$\G_N$-coset codes are a class of block codes \cite{arikan2009channel} that include polar codes\cite{stolte2002rekursive,arikan2009channel} and \ac{RM} codes\cite{reed,muller}. Polar codes achieve capacity over \acp{B-DMC} under low-complexity \ac{SC} decoding\cite{arikan2009channel} and \ac{RM} codes achieve capacity for \acp{BEC} under \ac{ML} decoding\cite{kudekar17}.
However, their performance under \ac{SC} decoding~\cite{stolte2002rekursive} is not competitive in the short- to moderate-length regime, e.g., from $128$ to $1024$ bits\cite{Coskun18:Survey}.\footnote{\ac{RM} codes become less suited for \ac{SC} decoding with an increasing length\cite{IU21,CP21}. Asymptotically, their error probability under \ac{SC} decoding is lower-bounded by $\nicefrac{1}{2}$\cite[Section X]{arikan2009channel}.} Significant research effort hast been put into approaching \ac{ML} performance by modifying the SC decoding schedule and/or improving the distance properties~\cite{Dumer04,Dumer06,Dumer:List06,NC12,miloslavskaya2014sequential,trifonov2018score,jeong2019sc,RMpolar14,Mondelli14,afisiadis2014low,chandesris2018dynamic,tal2015list,trifonov16,PCC_Polar16,FabregasSiegel17,HDM18,Ye19,Ivanov19,Yuan19,arikan2019sequential,CNP20,CP20,LGZ21,LZG19,YFV20,RBV20,LZL21,Rowshan19,KKK19,MV20,GEE20,GEE21,CP21,TG21}.

The idea of \emph{dynamic} frozen bits lets one represent any linear block code as a \emph{modified} $\G_N$-coset code\cite{trifonov16}. This concept unifies the concatenated polar code approach, e.g., with a high-rate outer \ac{CRC} code, to improve the distance spectrum of polar codes so that they can be decoded with low to moderate complexity\cite{tal2015list}.

This paper proposes \emph{\ac{SCOS}} decoding as a complexity-adaptive \ac{ML} decoder for modified $\G_N$-coset codes. An extension of the algorithm limits the worse-case complexity while still permitting near-\ac{ML} decoding. The decoder can be used for standalone $\G_N$-coset codes or for \ac{CRC}-concatenated $\G_N$-coset codes. In particular, numerical results show that \ac{RM} and RM-polar codes with dynamic frozen bits of block length $N\in\{128,256\}$, e.g., \ac{PAC} codes \cite{arikan2019sequential} and dRM-polar codes, perform within $0.25$ dB of the \ac{RCU} bound\cite{Polyanskiy10:BOUNDS} with an average complexity very close to that of \ac{SC} decoding at a \ac{FER} of $10^{-5}$ or below. For higher \acp{FER}, the gap to the \ac{RCU} bound is even smaller with higher complexity.

\subsection{Preview of the Proposed Algorithm}
\ac{SCOS} decoding borrows ideas from \ac{SC}-based flip\cite{afisiadis2014low,chandesris2018dynamic}, sequential\cite{fano1963heuristic,NC12,miloslavskaya2014sequential,trifonov2018score,jeong2019sc} and list decoders\cite{D74,fossorier,wu2007soft,tal2015list}. It is a tree search algorithm that flips the bits of valid paths to find a leaf with higher likelihood than other leaves, if such a leaf exists, and repeats until the \ac{ML} decision is found. The search stores a list of branches that is updated progressively while running partial \ac{SC} decoding by flipping the bits of the most likely leaf at each iteration. The order of the candidates follows the probability that they provide the \ac{ML} decision. \ac{SCOS} does not require an outer code (as for flip decoders) or parameter optimization for the performance vs. complexity trade-off (as for sequential decoders).

This paper is organized as follows. Section \ref{sec:prelim} gives background on the problem. Section \ref{sec:sc_os} presents the \ac{SCOS} algorithm. A lower bound on its complexity for \ac{ML} performance is described in Section \ref{sec:comp}.  Section \ref{sec:numerical} presents numerical results and Section \ref{sec:conclusions} concludes the paper.
\section{Preliminaries}
\label{sec:prelim}
We begin by introducing notation. Let $x^a$ be the vector $(x_1, x_2, \dots, x_a)$; if $a=0$, then the vector is empty. Given $x^N$ and a set $\mathcal{A}\subset [N]\triangleq\{1,\dots,N\}$, let $x_{\mathcal{A}}$ be the subvector $(x_i:i\in\mathcal{A})$. For sets $\mathcal{A}$ and $\mathcal{B}$, the symmetric difference is denoted $\mathcal{A}\triangle\mathcal{B}$ and an intersection set as $\mathcal{A}^{(N)}\triangleq\mathcal{A}\cap[N]$. Uppercase letters refer to \acp{RV} and lowercase letters to their realizations. A \ac{B-DMC} is written as $ W : \mathcal{X} \rightarrow \mathcal{Y} $, with input alphabet $ \mathcal{X} = \{0,1\} $, output alphabet $ \mathcal{Y} $, and transition probabilities $ W(y|x) $ for $ x \in \mathcal{X} $ and $ y \in \mathcal{Y} $. The transition probabilities of $N$ independent uses of the same channel are denoted as $ W^N(y^N|x^N) = \prod_{i=1}^{N} W(y_i|x_i) $. Capital bold letters refer to matrices, e.g., $\B_N$ denotes the $N\times N$ \emph{bit reversal} matrix \cite{arikan2009channel}.
\subsection{$\G_N$-coset Codes}\label{sec:polar_RM}

Consider the matrix $\G_{N} = \B_N\G_{2}^{\otimes n}$, where $N = 2^n$ with a non-negative integer $n$ and $\G_{2}^{\otimes n}$ is the $n$-fold Kronecker product of $\G_{2}$ defined as
\begin{equation}
\G_2 \triangleq 
\begin{bmatrix}
1 & 0 \\
1 & 1 
\end{bmatrix}.
\label{eq:transform_matrix}
\end{equation}
For the set $\mathcal{A}\subseteq[N]$ with $|\mathcal{A}|=K$, let $U_{\mathcal{A}}$ have entries that are \ac{i.i.d.} uniform \emph{information} bits, and let $U_{\mathcal{A}^c}=u_{\mathcal{A}^c}$ be fixed or \emph{frozen}. The mapping $c^N = u^N\G_{N}$ defines a $\G_N$-coset code\cite{arikan2009channel}. Polar and \ac{RM} codes are $\G_N$-coset codes with to different selections of $\mathcal{A}$\cite{arikan2009channel,stolte2002rekursive}.

Using $\G_{N}$, the transition probability from $u^N$ to $y^N$ is $W_N(y^N|u^N) \triangleq W^N(y^N|u^N\G_{N})$. The transition probabilities of the $i$-th \emph{bit-channel}, an artificial channel with the input $u_i$ and the output $(y^N,u^{i-1})$, are defined by
\begin{equation}
	W_N^{(i)}(y^N,u^{i-1}|u_i) \triangleq \sum_{u_{i+1}^N\in\mathcal{X}^{N-i}}\frac{1}{2^{N-1}}W_N(y^N|u^N).
\end{equation}
A $(N,K)$ polar code is designed by placing the $K$ most reliable bit-channels with indices $i\in[N]$ under the assumption that $U_i$, $i\in[N]$, are \ac{i.i.d.} uniform \acp{RV}, into the set $\mathcal{A}$. For a channel parameter, $\mathcal{A}$ can be found using density evolution \cite{arikan2009channel,Mori:2009SIT}. An $r$-th order \ac{RM} code of length-$N$ and dimension $K = \sum_{i=0}^{r}\binom{n}{i}$, where $0\leq r\leq n$, is denoted as RM$(r,n)$. Its set $\mathcal{A}$ consists of the indices, $i\in[N]$, with the Hamming weight at least equal to $n-r$ for the binary expansion of $i-1$. For both codes, one sets $u_i=0$ for $i\in\mathcal{A}^c$.

We make use of \emph{dynamic} frozen bits\cite{trifonov16}. A frozen bit is dynamic if its value depends on a subset of information bits preceding it; the resulting codes are called modified $\G_N$-coset codes. Dynamic frozen bits give better performance as the decoding algorithm approaches \ac{ML} decoding\cite{Yuan19,arikan2019sequential,CNP20,CP20,CP21,LGZ21} since the weight spectrum of the resulting code tends to improve as compared to the underlying code\cite{Yuan19,LZG19,YFV20,RBV20,LZL21,LGZ21}.

\subsection{Related Decoding Algorithms}\label{sec:decoding}

\subsubsection{Successive Cancellation Decoding}

Let $c^{N}$ and $y^N$ be the transmitted and received words, respectively. The \ac{SC} decoder mimics an \ac{ML} decision for the $i$-th bit-channel sequentially from $i=1$ to $i=N$  as follows. For $i\in\mathcal{F}$ set $\hat{u}_{i}$ to its (dynamic) frozen value. For $i\in\mathcal{A}$ compute the soft message $\ell\left(\hat{u}_1^{i-1}\right)$ defined as
\begin{equation}\label{eq:soft_message}
    \ell\left(\hat{u}_1^{i-1}\right)\triangleq\log\frac{P_{U_i|Y^NU^{i-1}}(0|y^N,\hat{u}^{i-1})}{P_{U_i|Y^NU^{i-1}}(1|y^N,\hat{u}^{i-1})}
\end{equation}
by assuming that the previous decisions $\hat{u}_{1}^{i-1}$ are correct and the frozen bits after $u_i$ are uniformly distributed. Now make the hard decision
\begin{equation}\label{eq:dec_fnc}
	\hat{u}_{i}=\begin{cases}
	0 & \text{if }\ell\left(\hat{u}_1^{i-1}\right)\geq 0\\
	1 & \text{otherwise}.
	\end{cases}
\end{equation}

\subsubsection{Successive Cancellation List Decoding}\label{sec:list_decoding}
\ac{SCL} decoding tracks several \ac{SC} decoding paths\cite{tal2015list} in parallel. At each decoding phase $i\in\mathcal{A}$, instead of making a hard decision on $u_i$, two possible decoding paths are continued in parallel threads, leading to up to $2^k$ decoding paths. The maximum number of paths implements \ac{ML} decoding but with an exponential number of decoding paths. To limit the complexity, one may keep up to $L$ paths at each phase. The reliability of a decoding path $\tilde{u}^{i}$ is quantified by a \emph{\ac{PM}} defined as~\cite{balatsoukas2015llr}
\begin{align}
    M\left(\tilde{u}^{i}\right)&\triangleq -\log P_{U^i|Y^N}\left(\tilde{u}^i|y^N\right)\label{eq:score_scl}\\
    &=M\left(\tilde{u}^{i-1}\right)+\log\left(1+e^{-\left(1-2\tilde{u}_i\right)\ell\left(\hat{u}_1^{i-1}\right)}\right)\label{eq:score_scl2}
\end{align}
where \eqref{eq:score_scl2} can be computed recursively using the \ac{SC} decoding with $M\left(\hat{u}^0\right)\triangleq 0$. 
At the end of $N$-th decoding phase, a list $\mathcal{L}$ of paths is collected. Finally, the output is the bit vector minimizing the \ac{PM}:
\begin{equation}
    \hat{u}^N = \argmin_{\tilde{u}^{N}\in\mathcal{L}} M\left(\tilde{u}^{N}\right).
\end{equation}

\subsubsection{Flip Decoding}\label{sec:flip_decoding}
An early erroneous bit decision may cause error propagation due to the serial nature of \ac{SC} decoding. The main idea of \ac{SCF} decoding \cite{afisiadis2014low} is to try to correct the first erroneous bit decision by sequentially flipping the unreliable decisions. This procedure requires an error-detecting outer code, e.g., a \ac{CRC} code.

The \ac{SCF} decoder starts by performing \ac{SC} decoding for the inner code to generate the first estimate $\tilde{u}^N$. If $\tilde{u}^N$ passes the \ac{CRC} test, it is declared as the output $\hat{u}^N=\tilde{u}^N$. If not, then the \ac{SCF} algorithm attempts to correct the bit errors at most $T$ times. At the $t$-th attempt, $t\in[T]$, the decoder finds the index $i_t$ of the $t$-th least reliable decision in $\tilde{u}^N$ according to the amplitudes of the soft messages \eqref{eq:soft_message}. The \ac{SCF} algorithm restarts the \ac{SC} decoder by flipping the estimate $\tilde{u}_{i_t}$ to $\tilde{u}_{i_t}\oplus 1$. The CRC is checked after each attempt. This decoding process continues until the \ac{CRC} passes or $T$ is reached.

Introducing a bias term to account for the reliability of the previous decisions enhances the performance~\cite{chandesris2018dynamic}. The improved metric is calculated as
\begin{align}
    Q(i) = \left|\ell\left(\tilde{u}^{i-1}\right)\right| + \sum_{\substack{j\in\mathcal{A}^{(i)}}}\frac{1}{\alpha}\log\left(1+e^{-\alpha \left|\ell\left(\tilde{u}^{j-1}\right)\right|}\right)
\end{align}
where $\alpha>0$ is a scaling factor.

\ac{SCF} decoding can be generalized to flip multiple bit estimates at once, leading to \ac{DSCF} decoding~\cite{chandesris2018dynamic}. The reliability of the initial estimates $\tilde{u}_\mathcal{E}$, $\mathcal{E}\subseteq\mathcal{A}$, is described by
\begin{align}\label{eq:dscf}
    \!\!\!\!Q(\mathcal{E}) \!=\! \sum_{i\in\mathcal{E}}\left|\ell\left(\tilde{u}^{i-1}\right)\right| \!+\! \sum_{\substack{j\in\mathcal{A}^{(i_\text{max})}}}\!\frac{1}{\alpha}\log\!\left(1+e^{-\alpha \left|\ell\left(\tilde{u}^{j-1}\right)\right|}\right)
\end{align}
where $i_{\text{max}}$ is the largest element in $\mathcal{E}$. The set of flipping positions is chosen as the one minimizing the metric \eqref{eq:dscf} and is constructed progressively.

\subsubsection{Sequential Decoding}\label{sec:seq_RM}
We review two sequential decoding algorithms, namely \ac{SCS} decoding~\cite{NC12,miloslavskaya2014sequential,trifonov2018score} and \ac{SC-Fano} decoding~\cite{jeong2019sc,arikan2019sequential}.

\ac{SCS} decoding stores the $L$ most reliable paths (possibly) with different length and discards the rest whenever the stack is full. At each iteration, the decoder selects the most reliable path and create two possible decoding paths based on this path. The winning word is declared once a path length becomes $N$.
\ac{SC-Fano} decoding deploys a Fano search \cite{fano1963heuristic} that allows backward movement in the decoding tree and that uses a dynamic threshold.

Sequential decoding compares paths of different lengths. However, the probabilities $P_{U^i|Y^N}\left(\tilde{u}^i|y^N\right)$, $\tilde{u}^i\in\{0,1\}^i$, cannot capture the effect of the path's length. This effect was taken into account first by~\cite{trifonov2018score}. Similar approach was used by \cite{jeong2019sc} to account for the expected error rate of the future bits as
\begin{align}
    S\left(\tilde{u}^i\right)&\triangleq -\log\frac{P_{U^i|Y^N}\left(\tilde{u}^i|y^N\right)}{\prod_{j=1}^{i} \left(1-p_{j}\right)}\label{eq:Score}\\
    &=M\left(\tilde{u}^{i}\right) + \sum_{j=1}^{i}\log\left(1-p_{j}\right)
\end{align}
where $p_{j}$ is the probability of the event that the first bit error occurred for $u_j$ in \ac{SC} decoding and $S\left(\tilde{u}^{0}\right)\triangleq 0$. The probabilities $p_{i}$ can be computed via Monte Carlo simulation\cite{stolte2002rekursive,arikan2009channel} or they can be approximated via density evolution\cite{Mori:2009SIT} offline.

\section{SC Ordered Search Decoding}
\label{sec:sc_os}

The \ac{SCOS} decoder starts by \ac{SC} decoding to provide an output $\tilde{u}^N$ as the current most likely leaf. The initial \ac{SC} decoding computes and stores the \ac{PM} \eqref{eq:score_scl} and the score \eqref{eq:Score} associated to the flipped versions of the decisions $\tilde{u}_i$, $\forall i\in\mathcal{A}$, i.e., $M\left(\left(\tilde{u}^{i-1},\tilde{u}_i\oplus 1\right)\right)$ and $S\left(\left(\tilde{u}^{i-1},\tilde{u}_i\oplus 1\right)\right)$, respectively. Every index $i\in\mathcal{A}$ with $M\left(\left(\tilde{u}^{i-1},\tilde{u}_i\oplus 1\right)\right) < M(\tilde{u}^N)$ is a flipping set.\footnote{Each set is a singleton at this stage.} The collection of all flipping sets forms a list $\mathcal{L}$. Each list member is visited in ascending order according to the score associated with it.

Suppose that $\mathcal{E}^*=\argmin_{\mathcal{E}\in\mathcal{L}} S\left(\mathcal{E}\right)$, where $S\left(\mathcal{E}\right)$ is the score associated with the flipping set $\mathcal{E}$. The decoder returns to decoding phase $j\triangleq\min_{i\in(\mathcal{E}^*\triangle\mathcal{E}^*_{(\text{p})})} i$, where $\mathcal{E}^*_{(\text{p})}$ is the flipping set chosen at the previous iteration, which is initialized as the empty set. The decision $\tilde{u}_j$ is flipped and \ac{SC} decoding continues. The set $\mathcal{E}^*$ is popped from the list $\mathcal{L}$. The \acp{PM} \eqref{eq:score_scl} and scores \eqref{eq:Score} are calculated again for the flipped versions for decoding phases with $i>j$, $i\in\mathcal{A}$, and the list $\mathcal{L}$ is enhanced by new flipping sets progressively (similar to \cite{chandesris2018dynamic}). The branch is discarded if at any decoding phase its \ac{PM} exceeds that of the current leaf, i.e., $M(\tilde{u}^N)$.\footnote{This pruning method is similar to the adaptive skipping rule proposed in~\cite{wu2007soft} for ordered-statistics decoding~\cite{D74,fossorier}.} Such a branch cannot output the \ac{ML} decision, since for any valid path $\tilde{u}^i$ the \ac{PM} \eqref{eq:score_scl2} is non-decreasing for the next stage, i.e., we have
\begin{align}\label{eq:pm_increase}
	M\left(\tilde{u}^{i}\right) \leq M\left(\tilde{u}^{i+1}\right), \forall \tilde{u}_{i+1}\in\{0,1\}.
\end{align}
If a leaf with lower \ac{PM} is found then it replaces the current most likely leaf. The procedure is repeated until it is impossible to find a more reliable path by flipping decisions, i.e., until $\mathcal{L}=\varnothing$. Hence, the \ac{SCOS} decoding implements an \ac{ML} decoder.

\begin{example}\label{exp:exp_scos}
Consider the $\left(4,2\right)$ polar code with $\mathcal{A}=\left\{2,4\right\}$ and $p_{1}^4=\left(0.4512, 0.1813, 0.1813, 0.0952\right)$.
Suppose the channel \ac{LLR} vector is $\left(-1.2, +3.4, -2.2, +0.9\right)$.
The \ac{SCOS} decoder works as follows (visualized in Figure~\ref{fig:exp_scos}):
\begin{itemize}
	\item[1.] \ac{SC} decoding (black path) gives an initial valid path $\tilde{u}^4=\left(0,0,0,0\right)$ with $M\left(\tilde{u}^4\right)= 3.4$. During the \ac{SC} decoding, the metrics $M\left(\left(\tilde{u}^{i-1},\tilde{u}_i\oplus 1\right)\right)$ and $S\left(\left(\tilde{u}^{i-1},\tilde{u}_i\oplus 1\right)\right)$, $i\in\mathcal{A}$, are computed as
	\begin{align*}
        M\left(\left(0,\textcolor{blue}{1}\right)\right) = 2.1&\quad\text{and}\quad S\left(\left(0,\textcolor{blue}{1}\right)\right) = 1.3\\
    	M\left(\left(0,0,0,\textcolor{brown}{1}\right)\right) = 4.3&\quad\text{and}\quad S\left(\left(0,0,0,\textcolor{brown}{1}\right)\right)= 3.2.
	\end{align*}
	As $M\left(\left(0,\textcolor{blue}{1}\right)\right) < M\left(\tilde{u}^4\right)$, we have a list $\mathcal{L}=\{\{2\}\}$.
	\item[2.] The decoder turns back to decoding stage $j=2$ since $\mathcal{E}^*=\{2\}$, flips the decision for $u_2$ to $\textcolor{blue}{1}$ (blue path) and continues \ac{SC} decoding (red path). The set $\mathcal{E}^*$ is popped from the list $\mathcal{L}$. Following the red path, the output is $\left(0,\textcolor{blue}{1},\textcolor{red}{0,1}\right)$ with $M\left(\left(0,\textcolor{blue}{1},\textcolor{red}{0,1}\right)\right) = 2.1$. Since $M\left(\left(0,\textcolor{blue}{1},\textcolor{red}{0,1}\right)\right) < M\left(\tilde{u}^4\right)$, the initial decision is updated as $\tilde{u}^4 = \left(0,\textcolor{blue}{1},\textcolor{red}{0,1}\right)$. Similarly, $M\left(\tilde{u}^i\right)$ and $S\left(\tilde{u}^i\right)$ are computed, for $i > 2$, $i\in\mathcal{A}$, during the decoding as
	\begin{equation*}
	    M\left(\left(0,\textcolor{blue}{1},\textcolor{red}{0},\textcolor{cyan}{0}\right)\right) = 5.6\quad\text{and}\quad S\left(\left(0,\textcolor{blue}{1},\textcolor{red}{0},\textcolor{cyan}{0}\right)\right)  = 4.5
	\end{equation*}
	\item[3.]  As $M\left(\tilde{u}^4\right)<M\left(\left(0,\textcolor{blue}{1},\textcolor{red}{0},\textcolor{cyan}{0}\right)\right)$, the list is empty, i.e., $\mathcal{L}=\varnothing$. Hence, the decoding is terminated and the \ac{ML} decision is $\hat{u}^4 = \tilde{u}^4$.	
	\end{itemize}
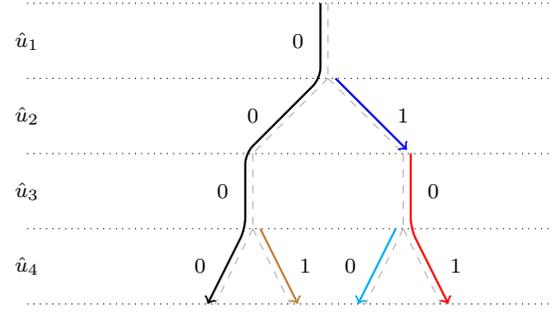
\begin{figure}
	\centering
	\footnotesize
	\begin{tikzpicture}
	\draw[black!30, dashed] (0,0) to (0,-1); \node at (-0.4,-0.5) {$0$};
	\draw[black!30, dashed] (0,-1) to (-1,-2); \node at (-1,-1.5) {$0$};
	\draw[black!30, dashed] (0,-1) to (1,-2); \node at (1,-1.5) {$1$};
	
	\draw[black!30, dashed] (-1,-2) to (-1,-3); \node at (-1.4,-2.5) {$0$};
	\draw[black!30, dashed] (1,-2) to (1,-3); \node at (1.4,-2.5) {$0$};
	\draw[black!30, dashed] (-1,-3) to (-1.5,-4); \node at (-1.7,-3.5) {$0$};
	\draw[black!30, dashed] (-1,-3) to (-0.5,-4); \node at (-0.3,-3.5) {$1$};
	\draw[black!30, dashed] (1,-3) to (1.5,-4); \node at (1.7,-3.5) {$1$};
	\draw[black!30, dashed] (1,-3) to (0.5,-4); \node at (0.3,-3.5) {$0$};

	\draw[dotted] (-4,0) to (3,0);
	\draw[dotted] (-4,-1) to (3,-1);
	\draw[dotted] (-4,-2) to (3,-2);
	\draw[dotted] (-4,-3) to (3,-3);
	\draw[dotted] (-4,-4) to (3,-4);
	
	\node at (-4,-0.5) {$\hat{u}_1$};
	\node at (-4,-1.5) {$\hat{u}_2$};
	\node at (-4,-2.5) {$\hat{u}_3$};
	\node at (-4,-3.5) {$\hat{u}_4$};
	
	\draw[->, rounded corners, thick] (-0.1,0) -- (-0.1,-1) -- (-1.1,-2) -- (-1.1,-3) -- (-1.6,-4);
	
	\draw[->, rounded corners, brown, thick] (-0.9,-3) -- (-0.4,-4);
	\draw[->, rounded corners, blue, thick] (0.1,-1) -- (1.05,-1.95);
	\draw[->, rounded corners, red, thick] (1.1,-2) -- (1.1,-3) -- (1.6,-4);
	\draw[->, rounded corners, cyan, thick] (0.9,-3) -- (0.4,-4);
	
\end{tikzpicture}
	\caption{The \ac{SCOS} decoding tree for an $\left(4,2\right)$ polar code.}
	\label{fig:exp_scos}
\end{figure}
\end{example}
\section{On the Complexity for ML Performance}
\label{sec:comp}

The decoding complexity is measured by the number of node-visits. For instance, the number of node-visits for \ac{SC} decoding $\chi_\text{SC}$ is simply the code length $N$. On the other hand, the complexity of \ac{SCOS} decoding is a \ac{RV} denoted as $\mathrm{X}$. 
\begin{remark}
    \ac{SCOS} decoding may visit the same node more than once and these visits are included in the comparison. To understand the minimum required complexity for \ac{SCOS} decoding, we define the set of partial input sequences $\tilde{u}^i$ with $i\in[N]$ with a smaller \ac{PM} than the \ac{ML} decision $\hat{u}^N_{\text{ML}}$.\footnote{There are $i$ node-visits for \ac{SC} decoding for any decoding path $\tilde{u}^i$.}
\end{remark}
\begin{definition}\label{def:setV}
    Let $\hat{u}^N_{\text{ML}}$ be the \ac{ML} decision given $y^N$. Define the set
    \begin{equation}
        \mathcal{V}\left(\hat{u}^N_{\text{ML}},y^N\right)\triangleq\bigcup_{i=1}^N\left\{u^i\in\{0,1\}^i: M\left(u^i\right) \leq M\left(\hat{u}^N_{\text{ML}}\right)\right\}.\label{eq:set_def}
    \end{equation}
\end{definition}
\begin{lemma}\label{lem:comp}
    For a particular realization $y^N$, we have  
    \begin{equation}
        \chi\geq \left|\mathcal{V}\left(\hat{u}^N_{\text{ML}},y^N\right)\right|\label{eq:complexity_realization}
    \end{equation}
    and the expected complexity is lower bounded as
    \begin{equation}
        \frac{1}{\chi_\text{SC}}\mathbb{E}\left[\mathrm{X}\right]\geq \frac{1}{N}\mathbb{E}\left[\left|\mathcal{V}\left(\hat{u}^N_{\text{ML}}(Y^N),Y^N\right)\right|\right].\label{eq:complexity_average}
    \end{equation}
\end{lemma}
\begin{proof}
The inequality \eqref{eq:complexity_realization} follows from the definition \eqref{eq:set_def} and the description of \ac{SCOS} decoding in Section \ref{sec:sc_os}. Since \eqref{eq:complexity_realization} is valid for any $y^N$, the bound \eqref{eq:complexity_average} follows by $\chi_\text{SC}=N$.
\end{proof}
\begin{remark}
    Recall that the \ac{PM} \eqref{eq:score_scl} is calculated using the \ac{SC} decoding schedule, i.e., it ignores the frozen bits coming after the current decoding phase $i$. This means the size of the set \eqref{eq:set_def} tends to be smaller for codes more suited for \ac{SC} decoding, e.g., polar codes, while it gets larger for others, e.g., \ac{RM} codes. This principle is also observed when decoding via \ac{SCL} decoding, i.e., the required list size to get close to \ac{ML} performance gets larger when one ``interpolates'' from polar to \ac{RM} codes\cite{Mondelli14,RMpolar14,CNP20,CP21}. This motivates introducing dRM-polar codes in Section~\ref{sec:numerical} that provide a good performance vs. complexity trade-off under \ac{SCOS} decoding for moderate code lengths, e.g., $N=256$ bits.
\end{remark}
\begin{remark}
    \ac{SCOS} decoding can be extended by choosing a maximum complexity $\chi_\text{max}/N$. This modification is useful for low \acp{SNR}, but of course the decoder is no longer \ac{ML} in general. To also limit the space complexity, one can limit the list size, e.g., we use $\left|\mathcal{L}\right| \leq \log_2 N \times \chi_\text{max}/N$ for the simulations in the next section.
\end{remark}
\section{Numerical Results}
\label{sec:numerical}
This section provides simulation results for \ac{biAWGN} channels. We compute \acp{FER} and complexities for modified $\G_N$-coset codes under \ac{SCOS} decoding with a different maximum number of node visits. The \ac{RCU} and \ac{MC} bounds\cite{Polyanskiy10:BOUNDS} are plotted as benchmarks. The empirical \emph{ML lower bounds} of~\cite{tal2015list} are also plotted: for \ac{SCOS} decoding with the largest maximum complexity constraint, each time a decoding failure occurred the decision $\hat{u}^N$ was checked. If
\begin{align}
	M(\hat{u}^N) \leq M(u^N).
\end{align}
then even an \ac{ML} decoder would make an error.

Figure~\ref{fig:scos_dfpolar} shows the \ac{FER} and complexity vs. \ac{SNR} in $E_b/N_0$ for a $\left(128,64\right)$ \ac{PAC} code~\cite{arikan2019sequential} under \ac{SCOS} decoding with a different maximum number of node-visits. The complexity is normalized by the complexity of \ac{SC} decoding. The information set $\mathcal{A}$ is the same as that of the \ac{RM} code, and the polynomial of the convolutional code is given by $\bm{g}=\left(0,1,1,0,1,1\right)$. In other words, we use a modified \ac{RM} code with dynamic frozen bits with the following constraints:
\begin{align}
    u_i = u_{i-2}\oplus u_{i-3}\oplus u_{i-5}\oplus u_{i-6},~i\in\mathcal{F}~\text{and}~i>6.
\end{align}
Since the average complexity gets large for (near-)\ac{ML} decoding of RM codes with dynamic frozen bits, namely dRM codes\cite{CP21}, an ensemble of modified RM-polar codes\cite{RMpolar14} is introduced.

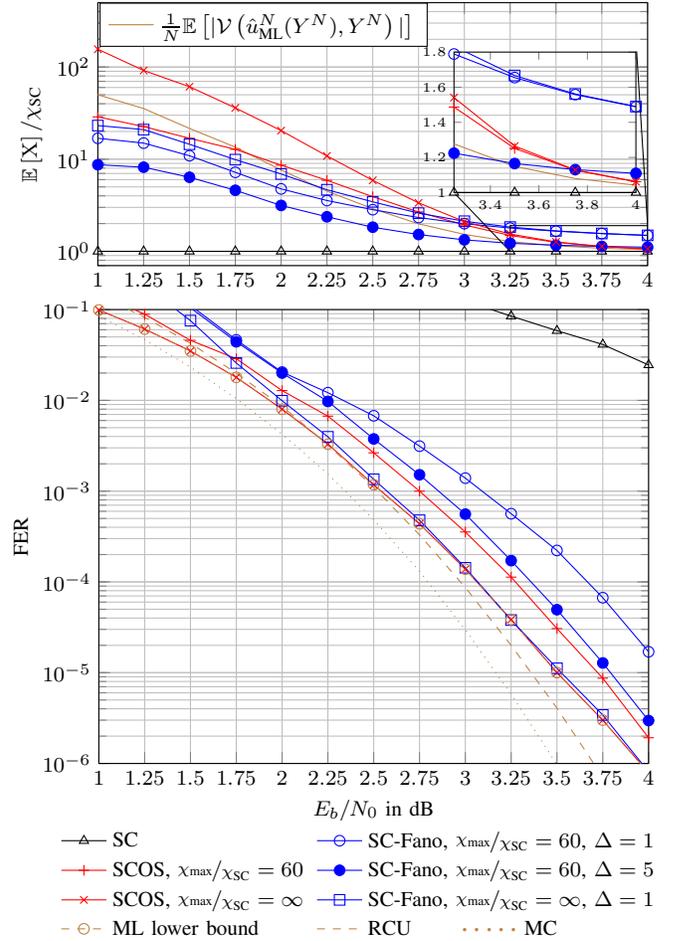
\begin{figure}
	\centering
	\footnotesize
	\begin{tikzpicture}[scale=1]
\footnotesize
\begin{semilogyaxis}[
mark options={solid},
width=3.5in,
height=2in,
legend style={at={(0.6,1)},anchor=north east},
ymin=0.7,
ymax=500,
grid=both,
xmin = 1,
xmax = 4,
ylabel = {$\mathbb{E}\left[\mathrm{X}\right] / \chi_\text{SC}$},
xtick=data
]

\addplot[brown]
table[x=snr,y=fer]{snr fer
	0 170.43
	0.25 132.12
	0.5 101.3
	0.75 76.784
	1 50.058
	1.25 35.447
	1.5 21.566
	1.75 13.517
	2 7.6471
	2.25 4.4946
	2.5 2.9054
	2.75 1.9874
	3 1.5279
	3.25 1.2769
	3.5 1.149
	3.75 1.079853
	4.0 1.041982
};\addlegendentry{$\frac{1}{N}\mathbb{E}\left[|\mathcal{V}\left(\hat{u}^N_\text{ML}(Y^N),Y^N\right)|\right]$}

\addplot[black, mark = triangle]
table[x=snr,y=fer]{snr fer
	0 1
	0.25 1
	0.5 1
	0.75 1
	1 1
	1.25 1
	1.5 1
	1.75 1
	2 1
	2.25 1
	2.5 1
	2.75 1
	3 1
	3.25 1
	3.5 1
	3.75 1
	4.0 1
};

\addplot[red, mark = +]
table[x=snr,y=fer]{snr fer
	0 51.474
	0.25 47.198
	0.5 42.789
	0.75 35.651
	1 28.807
	1.25 22.592
	1.5 16.862
	1.75 12.781
	2 8.6101
	2.25 5.8968
	2.5 3.9188
	2.75 2.6702
	3 1.9326
	3.25 1.487
	3.5 1.2515
	3.75 1.126594
	4 1.061820
};

\addplot[blue, mark = o]
table[x=snr,y=fer]{snr fer
  1.000000    16.856386 
   1.250000    14.949084
   1.500000    10.898122
   1.750000    7.200985 
   2.000000    4.759308 
   2.250000    3.559751 
   2.500000    2.830707 
   2.750000    2.326142 
   3.000000    1.997218 
   3.250000    1.789098 
   3.500000    1.654834 
   3.750000    1.557914 
   4.000000    1.488935
};

\addplot[blue, mark = *]
table[x=snr,y=fer]{snr fer
	0 2.771
	0.25 4.0333
	0.5 5.161
	0.75 7.3748
	1 8.7059
	1.25 8.1826
	1.5 6.3633
	1.75 4.5954
	2 3.1471
	2.25 2.3709
	2.5 1.8276
	2.75 1.5217
	3 1.3322
	3.25 1.2242
	3.5 1.1651
	3.75 1.130501
	4 1.109281
};

\addplot[blue, mark = square]
table[x=snr,y=fer]{snr fer
1.000000        23.15467
1.250000        20.95608
1.500000        14.51161
1.750000        9.887303
2.000000        6.906737
2.250000        4.659570
2.500000        3.446145
2.750000        2.596528
3.000000        2.121133
3.250000        1.833365
3.500000        1.666752
3.750000        1.561405
4.000000        1.489564
};

\addplot[red, mark = x]
table[x=snr,y=fer]{snr fer
	0 460.1
	0.25 372.04
	0.5 311.22
	0.75 211.98
	1 156.06
	1.25 91.848
	1.5 61.176
	1.75 36.022
	2 20.436
	2.25 10.851
	2.5 5.8956
	2.75 3.358
	3 2.103
	3.25 1.540764
	3.5 1.264679
	3.75 1.128713
	4 1.062140
};
\draw[black] (3.25,1) rectangle (4,1.9);
\draw (3.25,1) -- (2.95,4.2);
\draw (4,1.9) -- (3.94,148);
\coordinate (insetPosition) at (rel axis cs:1.0,0.19);
\end{semilogyaxis}
\begin{axis}[at={(insetPosition)},anchor={outer south east},tiny,xmin=3.25,xmax=4,ymin=1,ymax=1.8]

\addplot[brown]
table[x=snr,y=fer]{snr fer
	3.25 1.2769
	3.5 1.149
	3.75 1.079853
	4.0 1.041982
};

\addplot[black, mark = triangle]
table[x=snr,y=fer]{snr fer
	3.25 1
	3.5 1
	3.75 1
	4.0 1
};

\addplot[red, mark = +]
table[x=snr,y=fer]{snr fer
	3.25 1.487
	3.5 1.2515
	3.75 1.126594
	4 1.061820
};

\addplot[blue, mark = o]
table[x=snr,y=fer]{snr fer
   3.250000    1.789098 
   3.500000    1.654834 
   3.750000    1.557914 
   4.000000    1.488935
};

\addplot[blue, mark = *]
table[x=snr,y=fer]{snr fer
	3.25 1.2242
	3.5 1.1651
	3.75 1.130501
	4 1.109281
};

\addplot[blue, mark = square]
table[x=snr,y=fer]{snr fer
3.250000        1.833365
3.500000        1.666752
3.750000        1.561405
4.000000        1.489564
};

\addplot[red, mark = x]
table[x=snr,y=fer]{snr fer
	3.25 1.540764
	3.5 1.264679
	3.75 1.128713
	4 1.062140
};

\end{axis}
\end{tikzpicture}

\begin{tikzpicture}[scale=1]
\begin{semilogyaxis}[
legend style={at={(0.475,-0.125)},anchor=north,draw=none,/tikz/every even column/.append style={column sep=1mm},cells={align=left}},
legend columns=2,
legend cell align=left,
ymin=0.000001,
ymax=0.1,
width=3.5in,
height=3in,
grid=both,
xmin = 1,
xmax = 4,
xlabel = $E_b/N_0$ in dB,
ylabel = FER,
xtick={1,1.25,1.5,1.75,2,2.25,2.5,2.75,3,3.25,3.5,3.75,4}
]

\addplot[black, mark = triangle]
table[x=snr,y=fer]{snr fer
	0 0.88525
	0.25 0.85375
	0.5 0.7875
	0.75 0.74225
	1 0.6795
	1.25 0.60425
	1.5 0.528
	1.75 0.446
	2 0.361
	2.25 0.296
	2.5 0.23963
	2.75 0.17363
	3 0.1245
	3.25 0.08475
	3.5 0.05855
	3.75 0.041273
	4 0.024545
};\addlegendentry{SC}

\addplot[blue, mark = o]
table[x=snr,y=fer]{snr fer
       1.000000        0.351545 
        1.250000        0.227091
        1.500000        0.111000
        1.750000        0.046364
        2.000000        0.020500
        2.250000        0.012152
        2.500000        0.006764
        2.750000        0.003131
        3.000000        0.001391
        3.250000        0.000564
        3.500000        0.000222
        3.750000        0.000067
        4.000000        0.000017
};\addlegendentry{\ac{SC-Fano}, $\nicefrac{\chi_\text{max}}{\chi_\text{SC}}=60$, $\Delta=1$}

\addplot[red, mark = +]
table[x=snr,y=fer]{snr fer
	0 0.646
	0.25 0.5315
	0.5 0.401
	0.75 0.2595
	1 0.1485
	1.25 0.089
	1.5 0.045833
	1.75 0.02875
	2 0.012812
	2.25 0.0067
	2.5 0.0026447
	2.75 0.001
	3 0.00035461
	3.25 0.00011287
	3.5 3.0628e-05
	3.75 8.69e-6
	4 1.92e-06
};\addlegendentry{\ac{SCOS}, $\nicefrac{\chi_\text{max}}{\chi_\text{SC}}=60$}

\addplot[blue, mark = *]
table[x=snr,y=fer]{snr fer
	0 0.853
	0.25 0.7825
	0.5 0.679
	0.75 0.563
	1 0.398
	1.25 0.215
	1.5 0.1045
	1.75 0.044167
	2 0.02
	2.25 0.0097273
	2.5 0.0037593
	2.75 0.0015152
	3 0.00055833
	3.25 0.00017153
	3.5 4.9407e-05
	3.75 1.28e-05
	4 2.97e-06
};\addlegendentry{\ac{SC-Fano}, $\nicefrac{\chi_\text{max}}{\chi_\text{SC}}=60$, $\Delta=5$}

\addplot[red, mark = x]
table[x=snr,y=fer]{snr fer
	0 0.4405
	0.25 0.3175
	0.5 0.246
	0.75 0.1655
	1 0.099000
	1.25 0.061000
	1.5 0.035000
	1.75 0.017933
	2 0.008000
	2.25 0.003292
	2.5 0.001176
	2.75 0.00043668
	3 0.000138
	3.25 3.8432e-05
	3.5 1.0009e-05
	3.75 3.0e-06
	4 7.9e-07
};\addlegendentry{\ac{SCOS}, $\nicefrac{\chi_\text{max}}{\chi_\text{SC}}=\infty$}

\addplot[blue, mark = square]
table[x=snr,y=fer]{snr fer
1.000000        0.333273  
1.250000        0.185091  
1.500000        0.075545  
1.750000        0.025818  
2.000000        0.009909  
2.250000        0.003964  
2.500000        0.001351  
2.750000        0.000478  
3.000000        0.000143  
3.250000        0.000038  
3.500000        0.00001118
3.750000        0.00000342
4.000000        0.00000081
};\addlegendentry{\ac{SC-Fano}, $\nicefrac{\chi_\text{max}}{\chi_\text{SC}}=\infty$, $\Delta=1$}

\addplot[brown, dashed, mark = o, mark options={solid}]
table[x=snr,y=fer]{snr fer
	0 0.4405
	0.25 0.3175
	0.5 0.246
	0.75 0.1655
	1 0.099000
	1.25 0.061000
	1.5 0.035000
	1.75 0.017933
	2 0.008000
	2.25 0.003292
	2.5 0.001176
	2.75 0.00043668
	3 0.000138
	3.25 3.8432e-05
	3.5 1.0009e-05
	3.75 3.0e-06
	4 7.9e-07
};\addlegendentry{\ac{ML} lower bound}

\addplot[brown, dashed]
table[x=snr,y=fer]{snr fer
1.000000000000000   0.143839366949861
1.250000000000000   0.082213310092274
1.500000000000000   0.042497569653307
1.750000000000000   0.019950901922751
2.000000000000000   0.008586425015708
2.250000000000000   0.003223271460151
2.500000000000000   0.001100579891236
2.750000000000000   0.000330079099603
3.000000000000000   0.000086745193797
3.250000000000000   0.000019995279927
3.500000000000000   0.000004066498212
3.750000000000000   0.000000708069023
};\addlegendentry{RCU \qquad\textcolor{brown}{$\boldsymbol{\cdot\cdot\cdot\cdot\cdot}$} MC}

\addplot[brown, dotted]
table[x=snr,y=fer]{snr fer
1.000000000000000   0.084450067574090
1.250000000000000   0.046789395657259
1.500000000000000   0.023247677313516
1.750000000000000   0.010448670867371
2.000000000000000   0.004249278249449
2.250000000000000   0.001546310966604
2.500000000000000   0.000484410332805
2.750000000000000   0.000130900052036
3.000000000000000   0.000029657241148
3.250000000000000   0.000005796169358
3.500000000000000   0.000000972556776
};

\end{semilogyaxis}
\end{tikzpicture}
	\vspace*{-5mm}
	\caption{\ac{SCOS} decoding vs. \ac{SC-Fano} decoding for a $\left(128,64\right)$ \ac{PAC} code  with information set of \ac{RM}$(3,7)$ and polynomial $\bm{g}=\left(0,1,1,0,1,1\right)$.}
	\label{fig:scos_dfpolar}
\end{figure}

\begin{definition}
	The $(N,K)$ dRM-polar ensemble is the set of all codes, specified by the set $\mathcal{A}$ of an $(N,K)$ RM-polar code and choosing
	\begin{equation}
	    u_i = 0 \oplus \osum_{j\in\mathcal{A}^{(i-1)}} v_{j,i}u_j, \quad\forall i\in\mathcal{A}^c,
	\end{equation}
	with all possible $v_{j,i}\in\{0,1\}$ and $\mathcal{A}^{(0)}\triangleq\varnothing$, where $\osum$ denotes XOR summation.
\end{definition}

Figure~\ref{fig:scos_dfRMpolar} shows the simulation results for a rate $R\approx 0.6$ and length $N=256$ dRM-polar code chosen randomly from this ensemble.
\begin{figure}
	\centering
	\footnotesize
	\begin{tikzpicture}[scale=1]
\footnotesize
\begin{semilogyaxis}[
width=3.5in,
height=2in,
legend style={at={(1,1)},anchor=north east},
ymin=0.7,
ymax=800,
grid=both,
xmin = 1.5,
xmax = 3.5,
ylabel = {$\mathbb{E}|\left[\mathrm{X}\right] / \chi_\text{SC}$},
xtick={1.5,1.75,2.0,2.25,2.5,2.75,3.0,3.25,3.5},
]

\addplot[black, mark = triangle]
table[x=snr,y=fer]{snr fer
1.50   1
1.75   1
2.00   1
2.25   1
2.50   1
2.75   1
3.00   1
3.25   1
3.50   1
};

\addplot[blue, mark = o]
table[x=snr,y=fer]{snr fer
 1.50        83.795028   
 1.75        47.900659   
 2.00        21.651216   
 2.25        10.776071   
 2.50        5.446262    
 2.75        2.998036    
 3.00        2.114206    
 3.25        1.751856
 3.50        1.601123  
};

\addplot[blue, mark = square]
table[x=snr,y=fer]{snr fer
1.50        114.013499
1.75        72.945600 
2.00        34.447289 
2.25        16.217274 
2.50        6.972774  
2.75        3.402146  
3.00        2.192006  
3.25        1.763096
3.50        1.601308
};

\addplot[red, mark = +]
table[x=snr,y=fer]{snr fer
1.50   291.028451
1.75   169.699448
2.00   89.703434 
2.25   40.415621 
2.50   15.694348 
2.75   6.042126  
3.00   2.520955  
3.25   1.444602  
3.50   1.126563  
};

\addplot[red, mark = x]
table[x=snr,y=fer]{snr fer
1.50     764.952182 
1.75     405.065367 
2.00     181.191107 
2.25     71.323249   
2.50     24.959818  
2.75     7.772939  
3.00     2.839900  
3.25     1.479414 
3.50     1.130336
};

\addplot[blue, mark = *]
table[x=snr,y=fer]{snr fer
 1.50     47.521895 
 1.75      30.151412
 2.00      14.535406
 2.25      7.131651 
 2.50      3.323651 
 2.75      1.831700 
 3.00      1.337072 
 3.25      1.181829
 3.50      1.131066
};

\draw[black] (3.2,1) rectangle (3.5,1.9);
\draw (3.2,1) -- (2.768,15.5);
\draw (3.5,1.9) -- (3.43,660);
\coordinate (insetPosition) at (rel axis cs:1.0,0.35);
\end{semilogyaxis}
\begin{axis}[at={(insetPosition)},anchor={outer south east},tiny,xmin=3.2,xmax=3.5,ymin=1,ymax=1.9,xtick={3.2,3.3,3.4,3.5}]

\addplot[black, mark = triangle]
table[x=snr,y=fer]{snr fer
3.00   1
3.25   1
3.50   1
};

\addplot[blue, mark = o]
table[x=snr,y=fer]{snr fer
 3.00        2.114206    
 3.25        1.751856
 3.50        1.601123  
};

\addplot[blue, mark = square]
table[x=snr,y=fer]{snr fer
3.00        2.192006  
3.25        1.763096
3.50        1.601308
};

\addplot[red, mark = +]
table[x=snr,y=fer]{snr fer
3.00   2.520955  
3.25   1.444602  
3.50   1.126563  
};

\addplot[red, mark = x]
table[x=snr,y=fer]{snr fer
3.00     2.839900  
3.25     1.479414 
3.50     1.130336
};

\addplot[blue, mark = *]
table[x=snr,y=fer]{snr fer
  3.00      1.337072 
  3.25      1.181829
  3.50      1.131066
};

\end{axis}
\end{tikzpicture}

\begin{tikzpicture}[scale=1]
\begin{semilogyaxis}[
legend style={at={(0.45,-0.125)},anchor=north,draw=none,/tikz/every even column/.append style={column sep=1mm},cells={align=left}},
legend columns=2,
legend cell align=left,
ymin=0.000001,
ymax=0.1,
width=3.5in,
height=3in,
grid=both,
xmin = 1.5,
xmax = 3.5,
xlabel = $E_b/N_0\ \text{in dB}$,
ylabel = FER,
xtick={1.5,1.75,2.0,2.25,2.5,2.75,3.0,3.25,3.5},
]

\addplot[black, mark = triangle]
table[x=snr,y=fer]{snr fer
1.50        0.702273
1.75        0.590364
2.00        0.463727
2.25        0.344091
2.50        0.232727
2.75        0.155636
3.00        0.093091
3.25        0.045909
3.50        0.023455
3.75        0.011773
4.00        0.004659
};\addlegendentry{\ac{SC}}


\addplot[blue, mark = o]
table[x=snr,y=fer]{snr fer
1.50        0.106455  
1.75        0.040182  
2.00        0.012409  
2.25        0.004145  
2.50        0.001406  
2.75        0.000287  
3.00        0.000077  
3.25        0.00001331
3.50        0.00000209
};\addlegendentry{\ac{SC-Fano}, $\nicefrac{\chi_\text{max}}{\chi_\text{SC}}=10^3$, $\Delta=1$}

\addplot[red, mark = +]
table[x=snr,y=fer]{snr fer
1.50        0.066182
1.75        0.026818
2.00        0.010364
2.25        0.003394
2.50        0.000727
2.75        0.000178
3.00        0.000037
3.25        0.00000626
3.50        0.0000011
};\addlegendentry{\ac{SCOS}, $\nicefrac{\chi_\text{max}}{\chi_\text{SC}}=10^3$}

\addplot[blue, mark = *]
table[x=snr,y=fer]{snr fer
 1.50       0.137818 
 1.75        0.044273
 2.00        0.013061
 2.25        0.003935
 2.50        0.001098
 2.75        0.000224
 3.00        0.000048
 3.25        0.000008
 3.50        0.00000154
};\addlegendentry{\ac{SC-Fano}, $\nicefrac{\chi_\text{max}}{\chi_\text{SC}}=10^3$, $\Delta = 5$}

\addplot[red, mark = x]
table[x=snr,y=fer]{snr fer
1.50    0.055727  
1.75    0.022091  
2.00    0.007136  
2.25    0.001964  
2.50    0.000488  
2.75    0.000093  
3.00    0.000017  
3.25    0.00000295  
3.50    0.00000067
};\addlegendentry{\ac{SCOS}, $\nicefrac{\chi_\text{max}}{\chi_\text{SC}}=4\!\cdot\!10^3$}

\addplot[blue, mark = square]
table[x=snr,y=fer]{snr fer
1.50        0.105636
1.75        0.031455
2.00        0.008182
2.25        0.002443
2.50        0.000624
2.75        0.000122
3.00        0.0000252
3.25        0.00000451
3.50        0.00000078
};\addlegendentry{\ac{SC-Fano}, $\nicefrac{\chi_\text{max}}{\chi_\text{SC}}=4\!\cdot\!10^3$, $\Delta=1$}

\addplot[brown, dashed, mark = o, mark options={solid}]
table[x=snr,y=fer]{snr fer
1.50   0.051364
1.75   0.020273
2.00   0.006591
2.25   0.001655
2.50   0.000364
2.75   0.000063
3.00   0.00001338
3.25   0.00000248
3.50   0.00000061
};\addlegendentry{\ac{ML} lower bound}

\addplot[brown, dashed]
table[x=snr,y=fer]{snr fer
1.50   0.059757617047152
1.75   0.021787891578998
2.00   0.006164860529908
2.25   0.001447497327921
2.50   0.000267034704852
2.75   0.000038281156687
3.00   0.000004237053290
3.25   0.000000361836078
3.50   0.000000022242626
};\addlegendentry{RCU \qquad\textcolor{brown}{$\boldsymbol{\cdot\cdot\cdot\cdot\cdot}$} MC}

\addplot[brown, dotted]
table[x=snr,y=fer]{snr fer
1.50   0.038472876420095
1.75   0.013630064420632
2.00   0.003764763316368
2.25   0.000784435440668
2.50   0.000146821299603
2.75   0.000017424951529
3.00   0.000001717851313
3.25   0.000000114536883
3.50   0.000000006313519
};

\end{semilogyaxis}
\end{tikzpicture}
	\vspace*{-5mm}
	\caption{\ac{SCOS} decoding vs. \ac{SC-Fano} decoding for a $\left(256,154\right)$ dRM-polar code. The information set $\mathcal{A}$ is constructed as in \cite{RMpolar14} where the mother code is the RM$(4,8)$ and the polar rule is given by setting $\beta=2^{\nicefrac{1}{4}}$ in \cite{he2017beta}.}
	\label{fig:scos_dfRMpolar}
\end{figure}
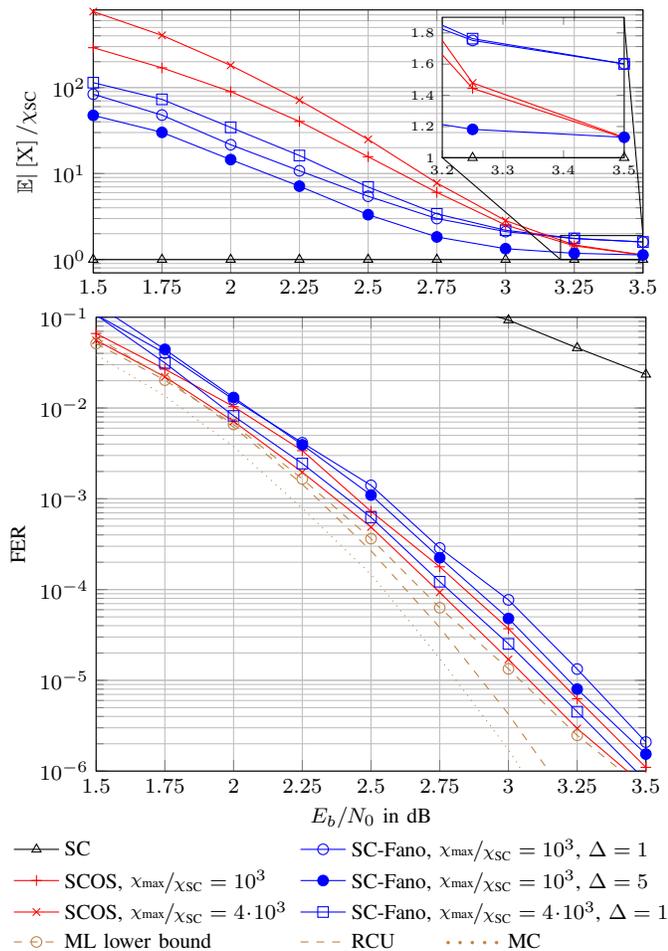
We observe the following behavior in Fig.~\ref{fig:scos_dfpolar} and Fig.~\ref{fig:scos_dfRMpolar}.
\begin{itemize}
	\item \ac{SCOS} decoding with unbounded complexity matches the \ac{ML} lower bound since it implements an \ac{ML} decoder.
	\item The average complexity $\mathbb{E}\left[\mathrm{X}\right]$ of a \ac{SCOS} decoder approaches the complexity of an \ac{SC} decoder for low \acp{FER} ($10^{-5}$ or below for the \ac{PAC} code and around $10^{-6}$ for the dRM-polar code). Indeed, it reaches the ultimate limit given by Lemma \ref{lem:comp}, which is not the case for \ac{SC-Fano} decoding. The difference to the \ac{RCU} bound\cite{Polyanskiy10:BOUNDS} is at most $0.2$ dB for the entire \ac{SNR} regime for the \ac{PAC} code and slightly larger for the dRM-polar code.
	\item The lower bound on the average given by \eqref{eq:complexity_average} is validated and is tight for high \ac{SNR}. However, the bound appears to be loose at low \ac{SNR} values mainly for two reasons: (i) usually the initial \ac{SC} decoding estimate $\tilde{u}^N$ is not the \ac{ML} decision and extra nodes in the difference set $\mathcal{V}\left(\tilde{u}^N,y^N\right)\setminus\mathcal{V}\left(\hat{u}^N_{\text{ML}},y^N\right)$ are visited and (ii) \ac{SCOS} decoding visits the same node multiple times and this cannot be tracked by a set definition. Considering (ii), it may be possible to reduce the number of revisits by improving the search schedule.
	\item A parameter $\Delta$ must be optimized carefully for \ac{SC-Fano} decoding to achieve \ac{ML} performance and this usually requires extensive simulations. Setting it small enough without any bound on the complexity would also practically achieve \ac{ML} performance; however, the complexity then explodes for longer codes. As seen from Figure~\ref{fig:scos_dfpolar}, $\Delta=1$ matches the \ac{ML} performance, but the average complexity is almost double that of \ac{SC} decoding near \acp{FER} of $10^{-5}$ or below. Moreover, under a maximum-complexity constraint, the average complexity of \ac{SC-Fano} decoding does not operate closer than \ac{SCOS} decoding to that of \ac{SC} decoding for similar performance.
	\item The parameter $\Delta$ must be optimized again for a good performance once a maximum-complexity constraint is imposed. Otherwise, the performance degrades significantly. Even so, \ac{SCOS} decoding outperforms \ac{SC-Fano} decoding for the same maximum-complexity constraint. However, \ac{SC-Fano} decoding has a lower average complexity for high \acp{FER} (if $\Delta$ is optimized) with a degradation in the performance. In contrast, \ac{SCOS} decoding does not require such an optimization.
\end{itemize}
\section{Conclusions}
\label{sec:conclusions}

A \ac{SCOS} decoding algorithm was proposed that implements \ac{ML} decoding. The complexity is adapted to the channel quality and approaches the complexity of \ac{SC} decoding for polar codes and short \ac{RM} codes at high \ac{SNR}. Unlike existing alternatives, the algorithm does not need an outer code or a separate parameter optimization. A lower bound on the complexity is approached for high \ac{SNR}. Finally, a code ensemble based on dRM-polar codes was introduced and a random instance performs within $0.25$ dB from the \ac{RCU} bound at a code length of $N=256$ bits with an average complexity close to that of \ac{SC} decoding.
	
\section*{Acknowledgements}
The authors thank Gerhard Kramer (TUM) for discussions which motivated the work and for improving the presentation. This work was supported by the German Research Foundation (DFG) under Grant KR~3517/9-1.


\end{document}